\documentclass[11pt]{article}
\tiny\normalsize

\usepackage{graphicx} 
\usepackage{booktabs} 
\usepackage{amsmath}
\usepackage{mathtools}
\usepackage{color, colortbl}
\usepackage{tikz}
\usepackage{multirow}
\usetikzlibrary{fit,positioning}
\usetikzlibrary{backgrounds}

\definecolor{blue}{rgb}{0,0,1}

\newtheorem{theorem}{Theorem}[section]

\newtheorem{definition}[theorem]{Definition}

\newtheorem{corollary}[theorem]{Corollary}

\newenvironment{proof}[1][Proof]{\begin{trivlist}
\item[\hskip \labelsep {\bfseries #1}]}{\end{trivlist}}

\newcommand{\qed}{\nobreak \ifvmode \relax \else
      \ifdim\lastskip<1.5em \hskip-\lastskip
      \hskip1.5em plus0em minus0.5em \fi \nobreak
      \vrule height0.75em width0.5em depth0.25em\fi}

\usepackage{float}
\usepackage[toc,page]{appendix}
\usepackage[noadjust]{cite}

\usepackage[colorlinks,citecolor=blue]{hyperref}
\usepackage[inline]{enumitem}
\usepackage[affil-it]{authblk}
\usepackage{physics}

\usepackage{mathtools}

\usepackage{enumitem}
\usepackage{algorithm}
\usepackage[noend]{algpseudocode}

\makeatletter
\def\BState{\State\hskip-\ALG@thistlm}
\makeatother

\graphicspath{
  {../Codes/Figures/}
}

\usepackage{amssymb}
\usepackage[square,sort,comma,numbers]{natbib}

\newcommand{\bitem}{\begin{itemize}}
\newcommand{\eitem}{\end{itemize}}
\newcommand{\benum}{\begin{enumerate}}
\newcommand{\eenum}{\end{enumerate}}
\newcommand{\beq}{\begin{equation}}
\newcommand{\eeq}{\end{equation}}
\newcommand{\beqs}{\begin{equation*}}
\newcommand{\eeqs}{\end{equation*}}

\newcommand{\be}{\begin{eqnarray}}
\newcommand{\ee}{\end{eqnarray}}
\newcommand{\ba}{\begin{eqnarray*}}
\newcommand{\ea}{\end{eqnarray*}}
\usepackage{enumitem}
\usepackage{enumitem}

\usepackage{algorithm}
\usepackage[noend]{algpseudocode}

\begin{document}
\bibliographystyle{cj}
\title{Multiple Hypothesis Testing To Estimate The Number Of Communities in Stochastic Block Models}
\author{Chetkar Jha and Mingyao Li and Ian Barnett}

\maketitle

\begin{abstract}
Clustering of single-cell RNA sequencing (scRNA-seq) datasets can give key insights into the biological functions of cells. Therefore, it is not surprising that network-based community detection methods are increasingly being used for the clustering of scRNA-seq datasets. The main challenge in implementing network-based community detection methods for scRNA-seq datasets is that these methods \emph{apriori} require the true number of communities or blocks for estimating the community memberships. Although there are existing methods for estimating the number of communities, they are not robust for noisy scRNA-seq datasets. 
Moreover, we require an appropriate method for extracting suitable networks from scRNA-seq datasets. For addressing these issues, we present a two-fold solution: i) a simple likelihood-based approach for extracting stochastic block models (SBMs) out of scRNA-seq datasets, ii) a new sequential multiple testing (SMT) method for estimating the number of communities in SBMs. Additionally, the main challenge with existing data analyses of scRNA-seq datasets is that they are sensitive to the choice of hyper-parameters. Our proposed approach also adresses this. 

We study the theoretical properties of SMT and establish its consistency under moderate sparsity conditions. We compare the numerical performance of the SMT with several existing methods. We also show that our approach performs competitively well against existing methods for estimating the number of communities on benchmark scRNA-seq datasets. Finally, we use our approach for estimating subgroups of a human retina bipolar single cell dataset. In the above instances, we never fix our hyper-parameters but rather select the best set of hyper-parameters through a grid-search.
\end{abstract}

\section{Introduction}
\subsection{Motivation}
Single-cell RNA sequencing (scRNA-seq) technologies have enabled us to collect large repositories of transcriptomic datasets. Variation in transcriptomic data at the cell level constitutes cellular heterogeneity. The main interest in molecular biology is to identify all observed cellular heterogeneities that are of biological relevance (Alstchuler and Wu (2010)\cite{altschuler2010}). A natural idea is to group these seemingly heterogeneous cells into groups of similar cells or clusters of cells. These cell clusters, in turn, can help us learn about the functional role of cells. However, the implementation of this scheme for scRNA-seq datasets is challenging due to the elevated noise. In particular, cells with low number of detected genes or contaminated cells constituting low-quality cells may increase the noise and obscure the biological signals thereby negatively affecting the clustering quality of scRNA-seq datasets.

Recently, several authors have proposed multiple approaches to tackle the problem of clustering of noisy scRNA-seq datasets. Broadly these ideas can be organized into the following steps:
i) filtering low-quality cells and genes, ii) log-normalizing the read counts, iii) reducing the number of features or performing dimension reduction, iv) estimating the number of communities, and v) clustering the transformed data; see Heumos et al. (2023) \cite{heumos2023} for a thorough discussion. The first step of filtering low-quality genes and cells are required for reducing the noise in scRNA-seq datasets whereas log-normalizing in the second step is required for boosting the signal strength while facilitating an even comparison of cells. Unfortunately, the first two steps cannot completely reduce the elevated noise level from scRNA-seq datasets and further steps are warranted. The existing methodology deals with the remaining noisiness of scRNA-seq datasets either through dimension reduction methods or by estimating the number of blocks. However, the existing approaches require fine-tuning of hyperparameters, see Shekhar et al. (2016) \cite{shekhar2016}, Kiselev et al. (2017) \cite{kiselev2017}. For instance, Louvain and Leiden algorithms implemented in Seurat (Satija et al. (2015) \cite{satija2015}), require the number of nearest neighbors beforehand to construct a weighted network for which they estimate the number of communities by modularity maximization approach (Newman and Girvan (2004) \cite{newman04}). These approaches require user-informed fine tuning of hyperparameters result in a clustering mechanism that is sensitive to hyperparameters.

A clustering process that divides the cells into multiple groups is extremely desirable as it helps us break-down the cell network into multiple homogeneous groups of similar cellular functions. This could give us insight into the latent biological process. This goal goes in hand with the the goal of developing clustering process that does not require frequent fine-tuning of hyperparameters and are fast algorithms. With these goals in mind, we study stochastic block models for scRNA-seq datasets.


\subsection{Stochastic Block Models} The stochastic block model (SBM) proposed by Holland et al. (1983) \cite{holland83} is one of the popular models for network data with the block or community structure. SBM simplifies the network model by assuming that the edge probability between any pair of nodes is solely determined by their respective community memberships. Although SBM has found favor in many application areas, it is still somewhat not adequate in the other application domains. Recently many extensions and variants of SBM were proposed. For instance, Airoldi et al. (2008) \cite{airoldi2008} proposed the mixed membership stochastic block model allowing nodes to belong to multiple networks. Karrer and Newman (2011) \cite{karrer2011} proposed degree corrected stochastic block models (DCSBMs) relaxing the assumption of homogeneity of nodes. Sischika and Kauermann (2025) \cite{sischika2025} combined a stochastic block model with a Graphon model. Currently, there are many existing methods for estimating the community structure of SBMs including modularity maximization (Newman and Girvan (2004) \cite{newman04}), Louvain modularity algorithms (Blondel et al. (2008) \cite{blondel2008}), likelihood-based approaches (Bickel and Chen (2009) \cite{bickel09}, Zhao et al. (2012) \cite{zhao12}, Choi et al. (2012) \cite{choi2012}, Amini et al. (2013) \cite{amini2013}), spectral clustering methods (Rohe et al. (2011) \cite{rohe11}, Lei and Ronaldo (2015) \cite{lei2015}, Joseph and Yu (2016) \cite{joseph2016}, Sengupta and Chen (2015) \cite{sengupta2015}) amongst others, see Zhao et al. (2017) \cite{zhao17} for a review. Several existing methods such as Newman and Girvan (2004) \cite{newman04}, Bickel and Chen (2009) \cite{bickel09}, Zhao et al. (2012) \cite{zhao12}, Qin and Rohe (2013) \cite{qin2013}, Amini et al. (2013) \cite{amini2013} have also shown to be consistent for both sparse and dense DCSBMs. However, invariably all of the existing community detection methods for SBMs require the true number of communities to be known beforehand.

\subsection{Estimating the Number of Communities in SBM} Fortunately, several methods have been proposed for estimating the number of communities in SBMs and DCSBMs. The existing methods could be grouped into five categories, namely: i) likelihood-based methods, ii) cross-validation based methods, iii) Bayesian methods, iv) semi-definite programming based methods, and v) spectral methods. Wang and Bickel (2017) \cite{wang17} proposed a likelihood ratio based method (LRBIC) that maximized a Bayesian information criterion to estimate the number of communities in SBMs. They use the likelihood ratio for performing a hypothesis test between $H_0: \hat{K}= K \text{ vs } H_1: \hat{K} > K$ which is sequentially used for estimating the number of communities. They established the consistency of LRBIC for the SBM as long as $K$ is fixed and the degree of the network (average number of edges per node) $d$ satisfy  $d \ge O(log(n))$. They extended their consistency result to DCSBM when the degree of the network $d$ was atleast of the order $O(n^{1/2}/ log(n))$. However, their approach is computationally expensive as it involves computing exponentially large number of terms. Ma et al. (2021) \cite{ma2019} fastened the above approach by maximizing a pseudo-likelihood ratio over several candidate number of communities for which the corresponding network was fit under a regularized spectral clustering method. They proved the consistency for their method for the fixed $K$ and semi-sparse regime where the tuning parameter $h_n \to 0, n \rho_n h_n \to \infty$ where $\rho_n$ is the sparsity parameter of the network. Motivated by Wang and Bickel (2017) \cite{wang17}, Jin et al. (2023) \cite{jin2023} proposed a new step-wise goodness of fit method (StGoF) to estimate the number of communities. The new GoF method has a phase transition property at the true $K$. Theoretically, their method is applicable for the case when the degree of the network $d \ge O(log(n))$ and $K$ is fixed. Compared to Wang and Bickel (2017) \cite{wang17}, the computation time of StGoF only increases quadratically $O(n^2)$ with $n$, which is still computationally expensive for large networks.

The cross-validation based approaches use network resampling strategies to generate multiple copies of the network. Subsequently, they use the cross-validation method to select the optimum number of communities. Particularly, Chen and Lei (2018) \cite{chen2018} and Li et al. (2020) \cite{li2020} proposed a network cross-validation (NCV) approach and an edge-cross validation approach (ECV) to estimate the number of communities. Li et  al. (2020)\cite{li2020}'s ECV approach is also valid for random dot product graph models. Theoretically, the two cross-validation approaches provide theoretical guarantees for not underestimating the number of communities when $K$ is fixed and the degree of the network satisfying $d \ge O(log(n))$. Overall, the cross-validation methods are computationally expensive and sensitive to overestimation problems.

Yan et al. (2018) \cite{yan2018} proposed a semi-definite relaxation method to estimate the number of communities in SBMs. Their objective function is the difference of the trace of the linear product of adjacency matrix with the clustering matrix and a tuning parameter multiplied with the trace of a clustering matrix. Theoretically, they proved that their method is consistent as long as the underlying signal (the difference within the block matrix probability and between block probability matrix) is large with $K$ being fixed and the degree of network satisfying $d \ge O(log(n))$. The requirement on the signal makes their method impractical for many real-life networks.

Cerqueira and Leonardi (2020) \cite{cerqueira2020} proposed a Bayesian approach for estimating the number of communities in SBMs while allowing the number of communities to increase with the network size $n$ at the maximum rate of $O(log(n))$. Recently, Cerqueira et al. (2023) \cite{cerqueira2023} extended \cite{cerqueira2020} to estimate the number of communities in DCSBMs under the semi-sparse regime, i.e., $\rho_n \to 0, n \rho_n \to \infty$. However, their method is computationally expensive involving computing integrated likelihood which makes their method unsuitable for large networks.
 
Spectral methods tend to rely on few extreme eigenvalues of a suitably modified adjacency matrix to estimate the number of blocks. Computing few extreme eigenvalues is computationally cheap even for large matrices making the spectral methods computationally efficient. One of the popular methods Lei (2016) \cite{lei16}'s Goodness-of-Fit (GoF) approaches uses the maximum and minimum eigenvalues of the generalized Wigner matrix to estimate a SBM network. Unlike majority of the existing methods, Lei (2016) \cite{lei16}'s GoF approach estimates the number of communities with $K$ increasing at the rate of $O(n^{1/6 
 - \tau})$ where $\tau > 0$. Unfortunately, GoF method is only applicable for dense SBM networks where the degree $d = O(n)$.
Other popular spectral approaches are Le and Levina (2022) \cite{le2022}'s BHMC and NB methods that utilize the spectrum of Bethe-Hessian matrix and Non-backtracking matrix respectively. Empirically, their methods are applicable for both sparse and dense SBMs/DCSBMs but theoretically they were only able to show the validity of their methods for a narrow region. They do allow for the number of communities to increase with $n$ but they do not give an analytical rate. Recently, Shao and Le \cite{le2024} proposed a non-backtracking matrix based approach for estimating the number of communities in sparse and imbalanced settings. However, they provided empirical and theoretical justification for only the Erd\"{o}s R\'{e}nyi case.
Also, Chen et al. (2023) \cite{chen2021}'s cross-validated eigenvalues use the idea that the sample eigenvectors under the null hypothesis should be orthogonal to the true latent dimensions. They adapt this idea to consistently estimate $K$ when the degree of the network $d \ge O(log(n))$. Their result holds for fixed $K$ where the $K$ largest eigenvalues satisfy some conditions related to the signal strength.

Most of the above methods view the problem of estimating the number of communities as a goodness of fit or model selection problem. Additionally, they make use of the entire network to estimate the number of communities. As we discuss, it will become clear that this does not come without a cost. In the present context of networks generated out of scRNA-seq datasets, the numerical simulations bear out the truth that we pay for using the methods in terms of computational time and/or comprising on accuracy when the out-in ratio is large. This is discussed later. 

Based on the preceding discussions, we propose a new approach for estimating the number of communities in sparse and dense SBMs. We require a method that has a satisfactory performance for sparse, noisy, and large networks. The high variability in noisy networks is partly due to the out-in ratio. The sensitivity of the out-in ratio is largely ignored in the statistical network literature. However, it is crucial for estimating the number of communities when the underlying network is highly sparse like in the present context. Our main observation is that a SBM consisting of $K$ blocks is equivalent to a  SBM consisting of $K$ distinct Erd\"{o}s R\'{e}nyi blocks. To avoid any ambiguity in identifying Erd\"{o}s R\'{e}nyi blocks in a SBM, we require that the edge probability with which edges are formed within Erd\"{o}s R\'{e}nyi blocks to be strictly greater than the edge probability with which edges are formed between a pair of Erd\"{o}s R\'{e}nyi blocks (assortativity). It immediately follows that estimating the number of blocks in a SBM is equivalent to estimating the number of Erd\"{o}s R\'{e}nyi blocks within the SBM.

Our main contributions are as follows: i) We propose a new approach for extracting SBM networks out of scRNA-seq datasets, ii) We propose a sequential multiple testing procedure (SMT) for estimating potentially large number of communities in semi-sparse/dense SBMs, iii) We prove that our estimator is consistent for estimating the true number of communities $K_{\star}^{(n)}$, iv) Numerically, we show that our method performs better even when the out-in ratio is large. 

\emph{Organization of the Paper}. The rest of the paper is organized as follows. Section \ref{prelim} introduces the basic set up by introducing the notations used in the paper while borrowing and defining key ideas in the paper. Particularly, we borrow a result on the limiting distribution of the second eigenvalue of appropriately scaled adjacency matrix. This section also introduces the SMT algorithm. Section \ref{theory} establishes the theoretical results involving SMT algorithm such as asymptotic null, asymptotic power, and consistency. Section \ref{evaluation} performs simulations and compares our method against other existing methods. Section \ref{data} introduces a new algorithm to extract SBM resembling networks out of scRNA-seq datasets. The section also performs the benchmark data analysis and real data analysis. Section \ref{discussion} summarizes the paper.

\section{Preliminaries}\label{prelim}
\subsection{Stochastic Block Model}
A stochastic block model (SBM) for a network consisting of $n$ nodes with $K_{\star}^{(n)}$ blocks is parametrized by a block membership vector $g$ and a symmetric block-wise edge probability matrix $G_{\rho_n}$. The block membership vector assigns nodes $\{1, \cdots, n \}$ to their respective blocks $\{1, \cdots, K_{\star}^{(n)} \}$ where the total number of blocks may vary with the block size $n$. 
Additionally, the SBM assumes that the probability of forming an edge between any pair of nodes $(i,j)$
can be completely characterized in terms of the probability of forming an edge between their respective blocks $g_i$ and $g_j$. The elements of $G_{\rho_n} = (\rho_n p_{ij})_{i,j =1, \cdots, K_{\star}^{(n)}}$ denotes the probability of forming an edge between the respective blocks $g_i$ and $g_j$. In particular, the mathematical relation between the node-wise edge probability matrix $P$ and the block-wise edge probability matrix $G_{\rho_n}$ is 

$$
P( A(i,j) =1 \mid g_i = k, g_j = k^{'} ) = G_{\rho_n}(k, k^{'}), 1 \le i,j \le n, 1 \le k, k^{'} \le K_{\star}^{(n)}.
$$

Additionally, we assume that $\{ p_{ij} \}_{i,j=1,\cdots, K_{\star}^{(n)} }$ do not depend on the network size $n$. The complexity of the SBM can be characterized in terms of the average degree of the network $d$ which is $O(n \rho_n)$. It is straightforward to observe that a network with larger degrees will have more observations for estimating the network parameters compared to a network with smaller degrees. In other words, the parameter estimation becomes a harder problem in sparse SBM networks.

Let $A = \{0,1\}^{n \times n}$ be the observed symmetric adjacency matrix with no self-loops (i.e., $A(i,i) =0$) for $1 \le i \le n$. Let every edge given $(g, G_{\rho_n})$ (up to a symmetric constant $A(i,j) = A(j,i)$) be an independent Bernoulli random variable, then the probability mass function of the adjacency matrix $A$ is
\be
P_{g, G_{\rho_n}}(A) = \prod_{1 \le i <j \le n} (G_{\rho_n}(g_i, g_j))^{A_{ij}} (1 - G_{\rho_n}(g_i,g_j))^{1 - A_{ij}}.
\ee

For the rest of the paper, we only consider the networks that are \emph{assortative} in the sense that the probability of forming an edge within blocks is greater than the probability of forming an edge between blocks. Additionally we assume that for any real symmetric Hermitian matrix of dimension $n \times n$, the eigenvalues are arranged in the descending manner $\lambda_1 > \lambda_2 > \cdots > \lambda_n$.

\subsection{Main Idea} Our approach is based on the observation that a SBM with $K_{\star}^{(n)}$ blocks consists of $K_{\star}^{(n)}$ distinct Erd\"{o}s R\'{e}nyi blocks. An Erd\"{o}s R\'{e}nyi graph is a random graph where edges between any distinct pair of nodes is assigned with a common probability $p$. For avoiding any ambiguity in identifying Erd\"{o}s R\'{e}nyi blocks in a SBM, we require that the edge probability within Erd\"{o}s R\'{e}nyi blocks to be strictly greater than the edge probability with which edges are formed between a pair of blocks (assortative graphs). Then, it immediately follows that estimating the number of blocks in a SBM is equivalent to estimating the number of distinct Erd\"{o}s R\'{e}nyi blocks within the SBM. We use this insight to propose the sequential multiple testing (SMT) method for estimating the number of blocks in a SBM.

Let $\{ A_{\star}^{(i)} = \{0, 1\}^{N_i \times N_i} \}_{i=1,\cdots,K}$ denote $K$ adjacency matrices corresponding to $K$ Erd\"{o}s R\'{e}nyi blocks with $N_i$ denoting the number of nodes belonging to $i^{th} $ Erd\"{o}s R\'{e}nyi block. Let $\{ p_{ii}\}_{i=1}^K$ denote the common node-wise probability with which the edges are formed in $K$ Erd\"{o}s R\'{e}nyi blocks. Using the above notations, we define the the scaled adjacency matrices $\{ M^{(i)}\}_{i=1}^K$ as follows

\be\label{scale.M}
M^{(i)}(u,v) = \begin{cases}
\frac{1}{\sqrt{N_i p_{ii} (1 - p_{ii})}} , & \text{ if } A^{(i)}_{\star}(u,v)=1, u,v =1, \cdots, N_i \\
0 , & \text{ if } A^{(i)}_{\star}(u,v) = 0, u, v =1 , \cdots, N_i,\\
\end{cases}
\ee\\

where 
$i=1, \cdots, K$.

\subsection{Second Eigenvalues of Erd\"{o}s R\'{e}nyi graphs} Notice that the first eigenvalue of an Erd\"{o}s R\'{e}nyi graph goes to infinity as the size of the Erd\"{o}s R\'{e}nyi graph increases. We can think of the first eigenvalue as a measure of signal. However, the second eigenvalue of the Erd\"{o}s R\'{e}nyi network when appropriately centered and scaled tend to remain stable which can be represented as the noise. We use this observation about the second eigenvalue of Erd\"{o}s R\'{e}nyi graph to identify an Erd\"{o}s R\'{e}nyi graph. To this end, we collect a pertinent result from \cite{leeTWCov16}.

\begin{theorem}\label{TW.thm}
Let $\epsilon > 0$, $N_i p_{ii} \ge N_i^{1/3}, \mu_i = N_i p_{ii}$ and $\tilde{\mu} = N_i (1 - p_{ii})$. Let $A^{(i)}_{\star}$ be the adjacency matrix generated from Erd\"{o}s R\'{e}nyi graph with $N_i$ nodes. Define $M^{(i)}$ as a scaled adjacency matrix defined in (\ref{scale.M}). Then the second eigenvalue of $M^{(i)}$ obeys the \emph{Tracy-Widom} distribution with a deterministic shift $\mu_i$, i.e., \\

\be\label{TW}
N_i^{2/3} (\lambda_2(M^{(i)}) - 2 -1/\mu_i) \to TW_1(\cdot),
\ee\\

where $TW_1(\cdot)$ is the Tracy-Widom distribution with Dyson parameter one.
\end{theorem}

Theorem \ref{TW.thm} characterizes Erd\"{o}s R\'{e}nyi graph using the limiting distribution of the second largest eigenvalue of the scaled adjacency matrix. In particular, the above theorem is valid when the average degree $d$ of the Erd\"{o}s R\'{e}nyi graph satisfies $d \ge \min_{i=1, \cdots, K} O(N_i^{1/3})$, where $N_i$ is the total number of nodes in the $i^{th}$ Erd\"{o}s R\'{e}nyi block. However, the above theorem is given in terms of unknown population parameter $G_{\rho_n}$ that requires estimation. As we discuss further, we give an estimable version of Theorem \ref{TW.thm}. Moreover, note that estimating $G_{\rho_n}$ would also require consistent recovery of the communities given the correct number of communities.

For recovering the true community membership, we collect results that guarantee recovering true community membership. Several methods can recover true true communities for the fixed $K_{\star}$ such as the profile likelihood method Bickel and Chen (2009) \cite{bickel09}, the spectral clustering method Lei and Zhu (2014) \cite{lei2014}. For $K_{\star}^{(n)}$ increasing with $n$, some methods can recover true communities for special cases such as planted partition models such as planted partition models see Chaudhuri et al. (2012) \cite{chaudhuri2012}, Amini and Levina (2018) \cite{amini2018}. Moreover for sparse SBMs, community detection methods such as Newman and Girvan (2004) \cite{newman04}, Bickel and Chen (2009) \cite{bickel09}, Zhao et al. (2012) \cite{zhao12}, Amini et al. (2013) \cite{amini2013} can also recover true community memberships.

\begin{definition}[(Consistency of Community Detection Methods)]\label{cons.def}
A sequence of SBMs (index with $n$) $\{ (g^{(n)}, G^{(n)}_{\rho_n}), n \ge 1 \}$ with $K_{\star}^{(n)}$ communities is said to have a consistent community membership estimator $\hat{g}(A, K_{\star}^{(n)})$, i.e.,\\
\ba
\lim_{n \to \infty} P(\hat{g} = g^{(n)} \mid A \sim (g^{(n)}, G^{(n)}_{\rho_n})) \to 1.
\ea
\end{definition}

\subsection{Sequential Multiple Tests} As discussed above, we observe that a SBM with $K_{\star}^{(n)}$ blocks consists of $K_{\star}^{(n)}$ Erd\"{o}s R\'{e}nyi blocks. For an Erd\"{o}s R\'{e}nyi graph with the common edge probability $p_{ii}$, we define the complement of the Erd\"{o}s R\'{e}nyi graph as follows. We connect the pair of nodes that were not connected originally and we disconnect the pair of nodes that had edges between them. Also, the complement of Erd\"{o}s R\'{e}nyi graph has no self-loops. It is easy to see that the complement of Erd\"{o}s R\'{e}nyi graph is also an Erd\"{o}s R\'{e}nyi with the common probability $(1 - p_{ii})$.

Let $\{ \hat{p}_{ii}\}_{i=1}^K$ be the common node-wise probability for $K$ Erd\"{o}s R\'{e}nyi blocks. Then, we define the estimated scale adjacency matrices and its complement as follows\\

\be\label{est.M}
\hat{M}^{(i)}(u,v) = \begin{cases}
\frac{1}{\sqrt{N_i \hat{p}_{ii}(1 - \hat{p}_{ii})}} , & \text{ if } A_{\star}^{(i)}(u,v)=1, u,v=1, \cdots, N_i, \\
0 , & \text{ if } A_{\star}^{(i)}(u,v) =0, u, v =1, \cdots, N_i
\end{cases}\\
\label{est.Mc}\hat{M}^{(i)}_c(u,v) = \begin{cases}
\frac{1}{\sqrt{N_{c, i} \hat{p}_{ii}(1 - \hat{p}_{ii})}} , & \text{ if } A_{\star}^{(i)}(u,v)=0, u,v=1, \cdots, N_{c, i}, \\
0 , & \text{ if } A_{\star}^{(i)}(u,v) =0, u, v =1, \cdots, N_{c,i}
\end{cases},
\ee\\
where $i=1, \cdots, K$., $N_i$ and $N_{c,i}$ denote the number of nodes in the $i^{th}$ Erd\"{o}s R\'{e}nyi block and its complement respectively.

For estimating the true number of communities, we construct a sequential test starting with $\hat{K}=1$ where $\hat{K}$ is incremented by one until the test is accepted. For $\hat{K} > 1$, we test whether $\hat{K}$ selected Erd\"{o}s R\'{e}nyi blocks are Erd\"{o}s R\'{e}nyi as well. However, it is evident that that keeping both candidate Erd\"{o}s R\'{e}nyi block and its complement for performing multiple testing is not useful as it would increase the type I error. Therefore, for every pair of candidate Erd\"{o}s R\'{e}nyi block and its complement, we only select one out of the pair for the purpose of multiple testing. Particularly, we select the $i^{th}$ candidate Erd\"{o}s R\'{e}nyi block or its complement depending on whether $\hat{G}_{\rho_n}(i,i) \le 0.5$ or $\hat{G}_{\rho_n}(i,i) > 0.5$. This is done in the above manner to optimize the power of the proposed sequential multiple test.

\begin{algorithm}[H]
\caption{Sequential Multiple Test Procedure}\label{alg:1}
\begin{algorithmic}[1]
\BState Initialize $\hat{K} =1$.
\BState\label{Step2} Use $\hat{K}$ to estimate community membership vector $\hat{g}$.
\BState Define adjacency matrices $A_{\star}^{(i)}$ for $i=1, \cdots, \hat{K}$.
\BState Using (\ref{est.M})-(\ref{est.Mc}), define estimated scaled adjacency matrices $M^{(i)}$ and $M^{(i)}_c$.
\BState Compute $\hat{G}_{\rho_n}$ as follows
\ba
\hat{G}_{\rho_n}(k,k') = \frac{\sum_{(s,t) : \hat{g}(s) =k, \hat{g}(t)= k' } A(s,t)}{\sum_{(s,t): \hat{g}(s)=k, \hat{g}(t) = k'} 1}, 1 \le k,k' \le \hat{K}.
\ea
\BState Compute the test statistics as follows
\be
T^{(i)}_{n, \hat{K}} = \begin{cases}
N_i^{2/3}(\lambda_2(\hat{M}^{(i)}) - 2 - \frac{1}{\mu_i}), & \text{ if } \hat{G}_{\rho_n}(i,i) \le 0.5 \\
N_{i}^{2/3} (\lambda_2(\hat{M})^{(i)} - 2 - \frac{1}{N_{i} - \mu_i}), & \text{ if } \hat{G}_{\rho_n}(i,i) > 0.5
\end{cases},
\ee
where $i =1, \cdots, \hat{K}$,
\be\label{t.stat}
T_{n, \hat{K}} = \max_{i=1, \cdots, \hat{K}} T^{(i)}_{n, \hat{K}},
\ee
where $\mu_i = N_i \star \hat{G}_{\rho_n}(i,i)$.
\BState For a specified significance level $\alpha$, conduct the multiple hypothesis test
\be\label{H0}
H_0: \text{ All $\hat{K}$ selected blocks are Erd\"{o}s R\'{e}nyi}, \\
\label{H1} H_1: \exists \text{ at least one block that is Erd\"{o}s R\'{e}nyi}.
\ee
\BState Accept $K = \hat{K}$ when $T_{n, \hat{K} \le TW_1(1-\alpha/\hat{K})}$ and Stop.
\BState If the test is rejected at the previous step then increment $\hat{K} = \hat{K} + 1$ and go to Step \ref{Step2}.
\end{algorithmic}
\end{algorithm}

\section{Main Results}\label{theory}
\subsection{Asymptotic Null} 
We obtain the estimable version of Theorem \ref{TW.thm} in terms of the estimated adjacency matrix. This incurs a cost in terms of the estimator error. We use Weyl's inequality to get a bound on the estimation error. Also, we assume that the blocks (under the null) are of similar sizes.

\begin{enumerate}[ label =A\arabic*]
\item\label{A1} (Balancedness) : Assume that all the communities of SBM are of similar sizes (under the null). 
\end{enumerate}

\begin{theorem} (Asymptotic Null Distribution)\label{est.TW.thm}
Let $A$ be an adjacency matrix generated from a SBM($g, G_{\rho_n}$) with $K_{\star}^{(n)}$ satisfying assumption (\ref{A1}). Assume that $\hat{g}$ is a consistent estimate of $g$. Let $\{ \hat{M}^{(i)}\}_{i=1}^{K_{\star}^{(n)}}$ be scaled adjacency matrices corresponding to the adjacency matrices $\{ A_{\star}^{(i)} \}_{i=1}^{K_{\star}^{(n)}}$. Then, we show that the second largest eigenvalue of the Erd\"{o}s R\'{e}nyi graph converges to the Tracy-Widom distribution under the following conditions

\be\label{est.tW}
(n/K_{\star}^{(n)})^{2/3}(\lambda_2(\hat{M}^{(i)}) - 2 - \frac{1}{\hat{\mu}_i}) \xrightarrow[]{D} TW_1(\cdot),
\ee
\text{ when } $\hat{G}_{\rho_n}(i,i) \ge (n/K_{\star}^{(n)}))^{-2/3},  \forall i =1, \cdots, K_{\star}^{(n)}$
for $K_{\star}^{(n)} = O(log(n))$.
\end{theorem}

\begin{proof}
The proof is given in the Supplement.
\end{proof}

Theorem \ref{est.TW.thm} states that the appropriately centered and scaled eigenvalue of the estimated scaled adjacency matrices corresponding to any $K_{\star}^{(n)}$ Erd\"{o}s R\'{e}nyi graph converges to the \emph{Tracy-Widom} distribution. The proof of the above theorem follows from bounding the estimation error using Weyl's inequality. Note that the above theorem puts an upper bound on $K_{\star}^{(n)}$ for the semi-sparse graph as $O(n)$. The above condition on $K_{\star}^{(n)}$ in (\ref{est.tW}) gets relaxed for the dense graphs to the $O(n^{1/4 - \tau})$.

\begin{corollary}\label{cor.est.TW}
For dense graphs when $\hat{G}_{\rho_n}(i,i) = O(1)$ with the growing number of communities $K_{\star}^{(n)} = O(n^{1/4 - \tau})$ with $\tau >0$, we show that
the second largest eigenvalue of the Erd\"{o}s R\'{e}nyi converges to the Tracy-Widom distribution under the following conditions
\be\label{est.tW}
(n/K_{\star}^{(n)})^{(2/3)}(\lambda_2(\hat{M}^{(i)}) - 2 - \frac{1}{\hat{\mu}_i}) \xrightarrow[]{D} TW_1(\cdot),
\ee
\text{ when } $\hat{G}_{\rho_n}(i,i) = O(1),  \forall i =1, \cdots, K_{\star}^{(n)}$.
\end{corollary}

\subsection{Asymptotic Power} The estimate of $K_{\star}^{(n)}$ is given as
\be\label{est.K}
\hat{K} = \inf_{k}\{k : k \in \mathbb{N} \ni T_{n, k} \le t_{1 - \alpha/k} \}.
\ee

For establishing the consistency of $\hat{K}$ in (\ref{est.K}), we need to show that the multiple hypothesis test in (\ref{H0})-(\ref{H1}) has the asymptotic power of one while we already bound the type I error using theorem \ref{est.TW.thm}.

\begin{theorem}\label{asm.power}
The power of the hypothesis test in testing $H_0$ vs $H_1$ (\ref{H0}- \ref{H1}) when the true number of blocks $K < K_{\star}^{(n)}$ follows the conditions outlined in theorem (\ref{est.TW.thm}) and corollary (\ref{cor.est.TW}), then
\be
P_{H_1}(T_{n, K} \ge t_{1 - \alpha/K}) \to 1,
\ee
where the test statistics $T_{n, K}$ is in (\ref{t.stat}).
\end{theorem}

Theorem \ref{asm.power} implies that we reject the null hypothesis whenever the model is under-fitted with probability one. For the under-fitted model, it means that there should be at least one selected Erd\"{o}s R\'{e}nyi block that is mis-aligned to the $j^{th}$ block. Remember that under the true model, the block sizes are balanced. Under the under-fitted model, we will always find a smaller block size of size $O(n/(K_{\star}^{(n)}\star K))$ that is mis-aligned to the $j^{th}$ block. This also means that there is an intra-graph between the $i^{th}$ block and $j^{th}$ block of size at least $O(n/(K_{\star}^{(n)}\star K))$ such that $p_{ii}^j$ (the edge probability between $i^{th}$ and $j^{th}$ block) is less than the $p^\star_{ii}$ (the true within edge probability of the $i^{th}$ block). The last inequality follows because of the assortative graphs. Summing up, for this smaller block, the edge probability has bumped up from $p^j_{ii}$ to $p^{\star}_{ii}$. This discrepancy allows us to measure the signal in terms of the difference between the test statistics computed under the under-fitted model and the null model. Moreover, this discrepancy or signal has a lower bound in terms of $\frac{\sqrt{p^\star_{ii} (1 - p^\star_{ii})}}{\sqrt{p_{ii}^j (1 - p_{ii}^j)}} - 1$ which is only positive for $p^\star_{ii} < 0.5$. For allowing $p^\star_{ii}$ to vary over the range $(0,1)$, we include the complement of Erd\"{o}s R\'{e}nyi graph. The rationale is that the edge block probability of Erd\"{o}s R\'{e}nyi graph or its complement will always be less than $1/2$ and therefore allowing us to measure the discrepancy. It is worth noting that in the algorithm \ref{alg:1}, we only pick one among the Erd\"{o}s R\'{e}nyi graph or its complement to construct a set of $K$ Erd\"{o}s R\'{e}nyi blocks. This is done to ensure the type I error is minimum while we optimize the asymptotic power.

\subsubsection{Investigations into Type I error and Power}
Our algorithm \ref{alg:1} and its theoretical results are also backed by the numerical performance. Here we investigate the type I error and power for the simulated SBMs when the degree, the network size $N$, and the number of communities $K$ is varied. Tables \ref{T1}-\ref{T2} compiles the power and type I error for different values of $N, K$, and degree. It is easy to see that as the degree decreases the type I error increases. However, it is under the theoretical bound at $O(N^{1/3})$. For even comparison with other methods, we have kept the out-in ratio as zero. However, we discuss the impact of larger out-in ratio in the next section.

\begin{table}[H]
\centering
\caption{Compilation of Power and Type I error: The table gives the numerical evaluation of power and type I error of SMT (\ref{alg:1}) for different values of $N,K$ when the out-in ratio is zero and power is $O(N^{2/3})$}\label{T1}.
\begin{tabular}{|c| c | c | c|}
\hline
N & K & Power & Type I error \\
\hline
500 & 2 & 1 & 0 \\
500 & 6 & 1 & 0 \\
\hline
1000 & 2 & 1 & 0 \\
1000 & 6 & 1 & 0 \\
\hline
5000 & 2 & 1 & 0 \\
5000 & 6 & 1 & 0\\
\hline
\end{tabular}
\end{table}

\begin{table}[H]
\centering
\caption{Compilation of Power and Type I error: The table gives the numerical evaluation of power and type I error of SMT (\ref{alg:1}) for different values of $N,K$ when the out-in ratio is zero and power is $O(N^{1/3})$}\label{T2}.
\begin{tabular}{|c| c | c | c|}
\hline
N & K & Power & Type I error \\
\hline
500 & 2 & 1 & 0 \\
500 & 6 & 1 & 0.02 \\
\hline
1000 & 2 & 1 & 0.01 \\
1000 & 6 & 1 & 0.03 \\
\hline
5000 & 2 & 1 & 0.01 \\
5000 & 6 & 1 & 0.05\\
\hline
\end{tabular}
\end{table}

\begin{corollary}\label{hatk}
    (Consistency of $\hat{K}$) The estimate given for $K_{\star}^{(n)}$ in algorithm \ref{alg:1} converges to $K_{\star}^{(n)}$ provided that the underlying SBM satisfies assumption \ref{A1} and one of the following condition is satisfied: i) $K_{\star}^{(n)} = O(n^{1/4 - \tau})$ and $G_{\rho_n}(i,i) = O(1)$, $\forall i=1, \cdots, K_{\star}^{(n)}$, ii) $K_{\star} = O(log(n)), G_{\rho_n}(i,i) = O( (n/K_{\star}^{(n)})^{-2/3})$ then $\hat{K}$ given in (\ref{est.K}) is consistent
    \be\label{hatk}
     \lim_{n \to \infty}   P(\hat{K} = K_{\star}^{(n)}) \to 1.
    \ee
\end{corollary}

Corollary (\ref{hatk}) guarantees the consistency of $\hat{K}$ in (\ref{est.K}) under the assumption of balancedness in (\ref{A1}) and some additional condition on the rate of $K_{\star}^{(n)}$ and the the edge probability $G_{\rho_n}$. Additionally, we assume that $\hat{g}$ can consistently recover the true community structure when the true $K$ is known, see Lei (2016) \cite{lei16} etc. The proof of the above corollary follows from the following observations. Since $\hat{K}$ is a sequential estimate starting with one, we have to establish that our algorithm does not under-fit. Using theorem \ref{asm.power}, we show that the probability of rejecting all under-fitted model is one. Moreover, theorem \ref{est.TW.thm} establishes that we can consistently estimate $K$ under the null, see the Supplement for the details

\subsection{Comparisons}
\subsubsection{Methodological Comparison}
For the readers who are familiar with Lei (2016) \cite{lei16}'s GoF method, the SMT algorithm \ref{alg:1} may seem eerily similar. However, the two methods are entirely different. The first major difference is that Lei (2016) \cite{lei16}'s GoF is a goodness of fit measure which uses the entire adjacency matrix for computing the test statistics. This is in contrast with the SMT where we perform a multiple hypothesis test to test on $K$ selected Erd\"{o}s R\'{e}nyi blocks which are a smaller part of the graph for $K > 1$. The second major difference is that Lei (2016) \cite{lei16} uses the first eigenvalue in constructing their test statistics while we utilize the second eigenvalue of Erd\"{o}s R\'{e}nyi graph. The intuition is that the second eigenvalue of Erd\"{o}s R\'{e}nyi graph is a measurement of noise and therefore can be used for a wider sparsity regime compared to the dense regime of the GoF method. The third major difference is that Lei (2016) \cite{lei16}'s GoF method is more sensitive to the out-in ratio compared to the SMT because our method only depends on $K$ major Erd\"{o}s R\'{e}nyi blocks. This feature of our method is particularly useful in the case of noisy and sparse scRNA-seq datasets which potentially can have larger out-in ratios. In fact, the above three points of difference can be used for contrasting the SMT with other existing methods. From our standpoint, all existing methods employ their methodology on the entire network for estimating the number of communities. As we discuss, this affects the performance of methods especially for sparse networks when the out-in ratio is large which is typical for scRNA-seq generated networks.

\subsubsection{Theoretical Comparisons}
It is also important to benchmark our method against other existing methods on important parameters that underpin theoretical guarantees. In statistical network literature, the performance of a method for estimating the number of communities is judged whether it can recover true number of communities in sparse settings $d > O(log(n))$. Since the existing version of our method is limited to SBM. We limit the comparisons to SBMs. Among likelihood methods, Wang and Bickel (2017) \cite{wang17}'s LRBIC and Jin et al. (2023) \cite{jin2023} can recover true number of communities for fixed $K$ in the sparse regime while Ma et al. (2021) \cite{ma2019}'s pseudo likelihood approach can estimate the true number of communities in the semi sparse regime. Purely on the degree criterion, these methods do well against SMT, but likelihood methods are computationally expensive and sensitive to the out-in ratio and estimates only fixed $K$. Cross validation methods such as Chen and Lei (2018) \cite{chen2018}, Li et al. (2020) \cite{li2020} provide theoretical guarantees for no underestimation for fixed $K$ in the sparse case $d > O(log(n))$. Theoretically, cross-validation methods do not provide stronger theoretical guarantees compared to SMT which guarantees asymptotic power of one for all under-fitted models. Additionally, SMT is computationally cheaper than the cross-validation methods. Yan et al. (2018) \cite{yan2018}'s semi-definite programming method can estimate the true number of communities for the sprase case. However, they require considerable large signals which makes their method impractical for non-trivial real life networks. Cerqueira and Leonardi (2020) \cite{cerqueira2020}'s Bayesian approach can estimate the true number of communities in the sparse case while allowing it to increase with $n$ at the rate of $O(log(n))$. Theoretically, their method fares better than SMT. However, their method can be computationaly expensive on the count of evaluating integral likelihood at every candidate value of $K$. Amongst spectral methods, Chen et al. (2023) \cite{chen2021}'s EigCV can estimate the true number of communities for fixed $K$ in the sparse case. However, they do make additional assumption on eigenvalues of networks which puts an additional assumption on the signal of the network, i.e., the out-in ratio. Theoretically, Lei (2016) \cite{lei16}'s GoF method is only applicable for the dense regime. Le and Levina (2022) \cite{le2022}'s BHMC provide theoretical guarantee on a narrow regime of $d$. Theoretically, SMT fares better than GoF and BHMC because SMT provides theoretical guarantee for a broder range of semi sparse networks while allowing the true number of communities to grow with $n$.

However, we feel that the theoretical comparison (on the basis of degree alone) is not a wholesome comparison as it does not consider the growing $K$, the computational cost, and the impact of the larger out-in ratio on the performance of these methods. The above considerations becomes important in the context of noisy scRNA-seq networks. Hence, we believe that it is important to select a method that is computationally efficient while handling larger out-in ratio. From our experiments, we believe that BHMC, EigCV, and SMT are the top three methods that satisfy the above requirement.

\section{Numerical Evaluation}\label{evaluation}
For comparison purpose, we consider the set of methods which are recent and where the authors have shared their code. For comparing the performance of our method against other competing methods, we consider the following settings. We generate SBMs by varying $K_{\star}$, the out-in ratio, and the degree of the network. In particular, we varied $K_{\star}$ over $(2,3,4)$, the out-in ratio over $(0, 0.2, 0.4, 0.6)$ and the degree over $(10, 20, 40, 80)$. The numerical performance of these methods are compared in terms of the accuracy rate. Tables \ref{T3}-\ref{T5} indicate the performance of the methods worsens as out-in ratio increases. In particular, it is easy to see that the best performance of all methods come under the scenario for low $K_{\star}$, small out-in ratio, and large degree.
As the number of community $K_{\star}$ increases or the degree decreases or the out-in ratio increases the numerical accuracy of all the methods worsen. This is to be expected.

Based on the numerical simulations, there is no single method that dominates. Therefore, we compare the performance of the methods in a pair-wise manner. The numerical results suggest that the performance of NCV and ECV are subpar compared to the performance of SMT and others. When SMT is compared against StGoF, we see that barring a couple of instances (when $K_{\star}= 2$, the out-in ratio is $0.6$, and the degree is either $10$ or $20$), SMT nearly dominates StGoF. In-between SMT and LRBIC, SMT's performance is nearly as good as LRBIC or better. Overall the numerical assessment suggests that SMT, BHMC, and EigCV are the top three methods with SMT being least sensitive to the increases in out-in ratios. 

We make further comparison between SMT, BHMC, and EigCV for large network and large $K$. Tables \ref{T6}-\ref{T7} compiles the numerical accuracy of SMT, BHMC, and EigCV when the true number of community is varied over $\{15, 20\}$ and the network size $N$ is fixed at $5000$. Here again, we see that SMT is less sensitive to the increases in out-in ratio. The less sensitivity of SMT with larger out-in ratio makes it a good candidate for the network analysis of noisy scRNA-seq data where the out-in ratio is expected to be larger.

\begin{table}[H]
\caption{Numerical Accuracy Comparison. The following table compiles the numerical accuracy of BHMC, LRBIC, NCV, ECV, StGoF, and SMT for various scenarios with the true number of communities fixed at $2$ and the network size at $500$.}\label{T3}
\centering
\begin{tabular}{|c| c | c | c| c | c | c | c | c| c| c|}
\hline
K & Out-in Ratio & Degree & BHMC & LRBIC & NCV & ECV & StGoF & EigCV & GoF & SMT 
\\
\hline
2 & 0 & 10 & 0.97 & 1 & 0.96 & 1 & 0.85 & 1 & 1 & 0.99 \\
2 & 0 & 20 & 1 & 1 & 0.99 & 0.99 & 0.85 & 1 & 1 & 0.99\\
2 & 0 & 40 & 1 & 1 & 1 & 1 & 0.01 & 1 & 1 & 0.99 \\
2 & 0 & 80 & 1 & 1 & 1 & 1 & 0.95 & 1 & 1 & 1 \\
\hline
2 & 0.2 & 10 & 0.99 & 1 & 0.98 & 1 & 0.92 & 1 & 0 & 1 \\
2 & 0.2 & 20 & 1 & 1 & 1 & 1 & 0.99 & 1 & 0.04 & 0.98\\
2 & 0.2 & 40 & 1 & 1 & 1 & 1 & 0.99 & 1 & 0.41 & 1 \\
2 & 0.2 & 80 & 1 & 1 & 1 & 1 & 1 & 1 & 0.61 & 0.99 \\
\hline
2 & 0.4 & 10 & 0.96 & 0 & 0.07 & 0.02 & 0.94 & 0.63 & 0 & 0.91 \\
2 & 0.4 & 20 & 1 & 1 & 1 & 1 & 0.98 & 1 & 0.23 & 0.98\\
2 & 0.4 & 40 & 1 & 1 & 1 & 1 & 1 & 1 & 0.66 & 0.99 \\
2 & 0.4 & 80 & 1 & 1 & 1 & 1 & 1 & 1 & 0.75 & 1 \\
\hline
2 & 0.6 & 10 & 0.04 & 0 & 0.04 & 0 & 0.85 & 0 & 0 & 0.12 \\
2 & 0.6 & 20 & 0.17 & 0 & 0.03 & 0 & 0.94 & 0.03 & 0.16 & 0.6\\
2 & 0.6 & 40 & 1 & 0.36 & 0.7 & 0.81 & 0.96 & 1 & 0.71 & 0.99 \\
2 & 0.6 & 80 & 1 & 1 & 1 & 1 & 0.99 & 1 & 0.86 & 1 \\
\hline
\end{tabular}
\end{table}

\begin{table}[H]
\caption{Numerical Accuracy Comparison. The following table compiles the numerical accuracy of BHMC, LRBIC, NCV, ECV, StGoF, and SMT for various scenarios with the true number of communities fixed at $3$ and the network size at $500$.}\label{T4}
\centering
\begin{tabular}{|c| c | c | c| c | c | c | c | c| c| c|}
\hline
K & Out-in Ratio & Degree & BHMC & LRBIC & NCV & ECV & StGoF & EigCV & GoF & SMT 
\\
\hline
3 & 0 & 10 & 1 & 1 & 0.96 & 1 & 0.05 & 1 & 1 & 1 \\
3 & 0 & 20 & 1 & 1 & 0.98 & 1 & 0.76 & 1 & 1 & 1\\
3 & 0 & 40 & 1 & 1 & 0.99 & 1 & 0.36 & 1 & 1 & 1 \\
3 & 0 & 80 & 1 & 1 & 0.99 & 1 & 0.34 & 1 & 1 & 1 \\
\hline
3 & 0.2 & 10 & 1 & 1 & 0.86 & 0.98 & 0.57 & 1 & 0 & 0.96 \\
3 & 0.2 & 20 & 1 & 1 & 1 & 1 & 0.98 & 1 & 0.04 & 1\\
3 & 0.2 & 40 & 1 & 1 & 0.98 & 1 & 1 & 1 & 0.33 & 1 \\
3 & 0.2 & 80 & 1 & 1 & 0.99 & 1 & 1 & 1 & 0.53 & 1 \\
\hline
3 & 0.4 & 10 & 0.01 & 0 & 0 & 0 & 0.01 & 0.63 & 0 & 0.02 \\
3 & 0.4 & 20 & 1 & 0.1 & 0.33 & 0.51 & 0.2 & 1 & 0.23 & 0.89\\
3 & 0.4 & 40 & 1 & 1 & 1 & 1 & 1 & 1 & 0.66 & 1 \\
3 & 0.4 & 80 & 1 & 1 & 1 & 1 & 1 & 1 & 0.75 & 1 \\
\hline
3 & 0.6 & 10 & 0 & 0 & 0 & 0 & 0.01 & 0 & 0 & 0.03 \\
3 & 0.6 & 20 & 0 & 0 & 0 & 0 & 0 & 0.03 & 0.16 & 0.04\\
3 & 0.6 & 40 & 0 & 0 & 0 & 0 & 0.01 & 1 & 0.71 & 0.12 \\
3 & 0.6 & 80 & 0.99 & 1 & 0.98 & 1 & 0.31 & 1 & 0.86 & 1 \\
\hline
\end{tabular}
\end{table}

\begin{table}[H]
\caption{Numerical Accuracy Comparison. The following table compiles the numerical accuracy of BHMC, LRBIC, NCV, ECV, StGoF, and SMT for various scenarios with the true number of communities fixed at $4$ and the network size at $500$.}\label{T5}
\centering
\begin{tabular}{|c| c | c | c| c | c | c | c | c| c| c|}
\hline
K & Out-in Ratio & Degree & BHMC & LRBIC & NCV & ECV & StGoF & EigCV & GoF & SMT 
\\
\hline
4 & 0 & 10 & 0.99 & 1 & 0.93 & 0.99 & 0.31 & 1 & 0.82 & 1 \\
4 & 0 & 20 & 1 & 1 & 0.96 & 1 & 0.57 & 1 & 1 & 1\\
4 & 0 & 40 & 1 & 1 & 1 & 0.99 & 0.67 & 1 & 1 & 1 \\
4 & 0 & 80 & 1 & 1 & 0.99 & 0.98 & 0 & 1 & 1 & 1 \\
\hline
4 & 0.2 & 10 & 0.99 & 0.77 & 0.27 & 0.48 & 0.11 & 0.88 & 0 & 0.75 \\
4 & 0.2 & 20 & 1 & 1 & 0.99 & 1 & 0.98 & 1 & 0.06 & 1\\
4 & 0.2 & 40 & 1 & 1 & 1 & 0.99 & 1 & 1 & 0.39 & 1 \\
4 & 0.2 & 80 & 1 & 1 & 0.96 & 0.99 & 1 & 1 & 0.52 & 1 \\
\hline
4 & 0.4 & 10 & 0 & 0 & 0 & 0 & 0 & 0 & 0 & 0 \\
4 & 0.4 & 20 & 0 & 0 & 0 & 0 & 0 & 0 & 0.25 & 0.15\\
4 & 0.4 & 40 & 1 & 1 & 0.91 & 0.98 & 0.37 & 1 & 0.68 & 0.94 \\
4 & 0.4 & 80 & 1 & 1 & 1 & 1 & 1 & 1 & 0.73 & 1 \\
\hline
4 & 0.6 & 10 & 0 & 0 & 0 & 0 & 0 & 0 & 0 & 0 \\
4 & 0.6 & 20 & 0 & 0 & 0 & 0 & 0 & 0 & 0.12 & 0\\
4 & 0.6 & 40 & 0 & 0 & 0 & 0 & 0 & 0 & 0.02 & 0 \\
4 & 0.6 & 80 & 0 & 0 & 0.09 & 0.09 & 0.01 & 0.03 & 0.32 & 0.59 \\
\hline
\end{tabular}
\end{table}

\begin{table}
\centering
\caption{Large K Comparison for Large Out-in Ratio: The table compiles the numerical accuracy of BHMC, EigCV, and SMT for large out-in ratio when $K =15$, $N = 5000$}\label{T6}
\begin{tabular}{|c | c | c | c | c | c |}
\hline
K & Degree & Out-In Ratio & BHMC & EigCV & SMT \\
\hline
15 & 100 & 0.2 & 1 & 1 & 0.82 \\
15 & 200 & 0.2 & 1 & 1 & 0.98 \\
15 & 400 & 0.2 & 1 & 1 & 1 \\
15 & 800 & 0.2 & 1 & 1 & 1 \\
\hline
15 & 100 & 0.4 & 0 & 0 & 0.16 \\
15 & 200 & 0.4 & 0 & 0.58 & 0.84 \\
15 & 400 & 0.4 & 1 & 1 & 0.94 \\
15 & 800 & 0.4 & 1 & 1 & 1 \\
\hline
15 & 100 & 0.6 & 0 & 0 & 0 \\
15 & 200 & 0.6 & 0 & 0 & 0.06 \\
15 & 400 & 0.6 & 0 & 0 & 0.14 \\
15 & 800 & 0.6 & 0 & 0 & 0.9 \\
\hline
\end{tabular}
\end{table}

\begin{table}
\centering
\caption{Large K Comparison for Large Out-in Ratio: The table compiles the numerical accuracy of BHMC, EigCV, and SMT for large out-in ratio when $K =20$, $N = 5000$}\label{T7}
\begin{tabular}{|c | c | c | c | c | c |}
\hline
K & Degree & Out-In Ratio & BHMC & EigCV & SMT \\
\hline
20 & 100 & 0.2 & 1 & 1 & 0.88 \\
20 & 200 & 0.2 & 1 & 1 & 0.86 \\
20 & 400 & 0.2 & 1 & 1 & 0.92 \\
20 & 800 & 0.2 & 1 & 1 & 0.96 \\
\hline
20 & 100 & 0.4 & 0 & 0 & 0.26 \\
20 & 200 & 0.4 & 0 & 0 & 0.38 \\
20 & 400 & 0.4 & 0 & 0.98 & 0.94 \\
20 & 800 & 0.4 & 1 & 1 & 0.98 \\
\hline
20 & 100 & 0.6 & 0 & 0 & 0.02 \\
20 & 200 & 0.6 & 0 & 0 & 0.04 \\
20 & 400 & 0.6 & 0 & 0 & 0.06 \\
20 & 800 & 0.6 & 0 & 0 & 0.34 \\
\hline
\end{tabular}
\end{table}

\section{Real Data Analysis}\label{data}
We adapt our method for performing clustering of the scRNA-seq datasets. This begs a question whether SBM is the best model for clustering of the scRNA-seq datasets? The case may be easily made that they are not the best model for the clustering of a particular scRNA-seq dataset. However, the SBM has clear advantage in being the simplest, the parsimonious and most explainable model for clustering of scRNA-seq datasets.
This has the added advantage that the parsimonious model tend to be more robust and when the data is as noisy as scRNA-seq dataset then there is a case for developing robust and parsimonious solutions. We believe that SBM based clustering method has the explainable power of decomposing cell networks into multiple homogeneous groups which would encourage many practitioners to adopt them. Additionally, we also validate our method by tracking the performance of our methods against existing state of art methods on the benchmark datasets. 
For the practical implementation, we note that Grabski, Street and Irizarry (2024)\cite{irizarry2024} remarked that most of the clustering methods of scRNA-seq are heuristic in nature and does not account for statistical uncertainty. We attempt to approach this by removing any fine-tuning of the hyperparameters and developing a model based clustering of the scRNA-seq datasets.

\emph{Extracting SBM networks from scRNA-seq datasets.} For applying SMT on a network generated out of scRNA-seq data, we have to carefully design a method to construct SBMs out of scRNA-seq networks. However, most of the data generating process frequently require some fine tuning of free hyperparameters, see Kiselev et al. (2017) \cite{kiselev2017}, Grabski et al. (2024) \cite{irizarry2024} among others. To this end, we perform data filtering and data normalization according to the best practices, see Heumos et al. (2023) \cite{heumos2023}. In general, we filter out cells that have less cell counts or more cell counts. Similarly, we filter out ubiquitous genes and/or rare genes. In particular, we vary the hyperparameter $\eta$ over a grid of values. Subsequently, we use the correlation matrix to construct a similarity matrix. Then, we construct an adjacency matrix by thresholding using the $\kappa^{th}$ quantile of the correlation matrix $R$,i.e.,
\ba
A_{\kappa, \eta} = \begin{cases}
1, & i \neq j, R(i,j) > q_{\kappa}(R_{\eta}), \\
0, & i \neq j, R(i,j) \le q_{\kappa}(R_{\eta}), \\
0, & i = j.
\end{cases}
\ea

For each adjacency matrix, we use SMT to estimate the true number of communities with $\alpha$ varying over $\{0.01, 0.05\}$. For convenience, we denote the estimated number of communities as $\hat{K}_{\alpha, \kappa, \eta}$. Then for every adjacency matrix we obtain the complement of the adjacency matrix. Subsequently for every combination of $\{\alpha, \kappa, \eta \}$, we compute the log likelihood under the SBM for each estimate $\hat{K}_{\alpha, \kappa, \eta}$. Finally, we select the network for which the log likelihood is maximized. This selected network would be the network that would closely resemble to a SBM.
The above approach is a standard approach for any model-based clustering. It is worth noting that practitioners may still chose other methods such as BHMC or EigCV instead of SMT. However, SMT has the comparative advantage that it is less sensitive to the out-in ratio making it a suitable candidate for the network analysis of scRNA-seq datasets.


\subsection{Benchmark Data Analysis} Kiselev et al. (2017) \cite{kiselev2017} identified six reference scRNA-seq datasets as \emph{gold standard} where there is a consensus on the underlying number of clusters. The six reference datasets can be downloaded from Gene Expression Omnibus (GEO) Edgar et al. (2022) \cite{edgar2002} with their ascension given in the Table \ref{T8}. Kiselev et al. (2017) \cite{kiselev2017}
pointed out that the ubiquitous genes and rare genes usually cannot help with the clustering. Hence, we filtered out genes that have less variability than $\eta^{th}$ quantile or more variability than $1 - \eta^{th}$ quantile where $\eta$ taking value in $\{0, 0.05, 0.1, 0.15, 0.2\}$. For every value of $\eta$, we get a subsetted data $z_{\eta}$. We filter out a fraction of genes that might affect the homogeneity of block structure of the resultant network. Then we log normalize the data using $log_2(1 + 10000 \star z_{\eta}(i,j)/\sum_i z_{\eta}(i,j) )$ where $z_{\eta}(i,j)$ is the $(i,j)^{th}$ of $z_{\eta}$. Moreover for every value of $\eta$, we compute the correlation matrix $R_{\eta}$ and use the thresholding method to compute the adjacency matrix as follows

\be
A_{\eta, \kappa}(i,j) = \begin{cases}
1 & R_{\eta}(i,j) > q_{\kappa}(R_{\eta}(i,)), i \neq j \\
0 & R_{\eta}(i,j) \le q_{\kappa}(R_{\eta}(i,)), i \neq j \\
0 & i=j,
\end{cases}
\ee
where $q_{\kappa}$ denote the $\kappa^{th}$ quantile.
For every adjacency matrix in $\{ A_{\eta, \kappa}\}$, we run SMT to obtain an estimate $\hat{K}_{\eta, \kappa, \alpha}$ and $\hat{g}_{\eta, \kappa, \alpha}$ with $\alpha$ varying over $\{0.05, 0.01 \}$. Using the preceeding estimates, we obtain the joint likelihood of adjacency matrix and it's complement assuming SBM. (The complement of adjacency matrix can be computed as we have computed the complement of Erd\"{o}s R\'{e}nyi block). The optimum value of $\eta, \kappa, \alpha$ are those for which the joint log likelihood is maximized. This means that for the estimated hyperparameters $\hat{\eta}, \hat{\kappa}, \hat{K}, \hat{\alpha}$ the adjacency matrix closely resembles SBM. Lastly, we compare our results against other results on benchmark  datasets. It is important to bear in mind that our method stand out in comparison to other methods because we do not require any fine-tuning of hyperparameters. Table \ref{T9} gives the comparative analysis on benchmark datasets. The table indicates that our method's performance is slightly better than SC3.

\begin{table}{h}
\centering
\caption{Summary of six reference datasets. The following is the summary level information on six reference datsets along with their GEO ascension numbers}\label{T8}
\begin{tabular}{| c | c | c | c | c |}
\hline
Datasets & Number of & Number of & Cell Resource & GEO ascension \\
 & Cells & Genes &  & number \\
\hline
Biase & 49 & 25,737 & 2-cell and 4-cell mouse embryos & GSE57249 \\
Yan & 124 & 22,687 & Human preimplantation embryos and & GSE36552 \\
 & & & embryonic stem cells & GSE36552 \\
Goolam & 124 & 41,480 & 4-cell mouse embryos & E-MTAB-3321 \\
Deng & 268 & 22,457 & Mammalian Cells & GSE47519 \\
Pollen & 301 & 23,730 & Human cerebral cortex & SRP041736 \\
Kolodziejczyk & 704 & 38,653 & Mouse embryonic stem cells & E-MTAB-2600\\
\hline
\end{tabular}
\end{table}

\begin{table}[h]
\centering
\caption{Comparison table for estimated number of communities for reference single cell datasets}\label{T9}
\begin{tabular}{|c | c | c | c | c | c|}
\hline
Datasets & References & SMT & SC3 & SIMLR & Seurat \\
\hline
Biase & 3 & 3  & 3 & 7 & 2 \\
Yan & 7 & 6 & 6 & 12 & 3 \\
Goolam & 5 & 5 & 6 & 19 & 4 \\
Deng & 10 & 11 & 9 & 16 & 7 \\
Pollen & 11 & 11 & 11  & 15 & 8 \\
Kolodziejcky & 3 & 10 & 10 & 6 & 7 \\
\hline
\end{tabular}
\end{table}

\subsection{Human Retina Bipolar Single-Cell Analysis} Bipolar cells in the eye play an important role in processing visual information. They connect the photoreceptor cells (input) and ganglion cells (output). The nature of bipolar cell signals and their transmission to ganglion cells is important to learn because it determines how ganglion cells integrate the visual information. Summarily, we do not know how bipolar cells process the complex visual input signals. A part of the difficulty is because bipolar cells are relatively inaccessible lying in between the photoreceptor cells and ganglion cells.

For the real data analysis, we use scRNA-seq data generated from the retina cells of two healthy adult donors. The scRNA-seq data was generated from the retina cells of two healthy adult donors using the 10X Genomics Chromium system. The detailed preprocessing and donor characteristics of our scRNA-seq data can be found in Lyu (2019) \cite{lyu2019}. In total, $33,694$ genes were sequenced over $92583$ cells. The data can be accessed using GSE155228. The scRNA-seq data consist of $33,694$ genes and $30,125$ bipolar cells. Subsequently, we processed the scRNA-seq data in line with the current best practices in the literature while optimizing the hyperparameters. First, we filtered out genes whose variability was less than top $\gamma$ genes. Second, we filtered out cells whose total cell counts (across all genes) were greater than $\delta^{th}$ percentile and less than $(1-  \delta)^{th}$ percentile. This helps filter out low quality cells that may have low quality and/or contaminated cells with large number of cells. Subsequently, we normalize the scRNA-seq data using $log(1 + \beta* x_{ij}/\sum_{j} x_{ij})$ with $x_{ij}$ denoting the $(i,j)^{th}$ entry of scRNA-seq data and $\sum_{j} x_{ij}$ denoting the row sum of the scRNA-seq data in the previous step. This helps us alleviate the noise while boosting the signal strength. Then we compute the correlation matrix $R$. For extracting a single cell network or an adjacency matrix, we use $R$ and the thresholding parameter $\kappa$ in the following manner
\ba
A_{\kappa}(i,j) = \begin{cases}
1, & R_{ij} \ge \kappa, i \neq j , \\
0, & R_{ij} < \kappa, i \neq j ,\\
0, & i =j.
\end{cases}
\ea

We parameterize the adjacency matrix as a function of hyperparameters $A_{\gamma, \delta, \beta, \kappa}$. In particular, we varied $\gamma$ over $\{2000, 4000, 6000\}$, $\delta$ over $\{0.1, 0.15, 0.2, 0.25, 0.3\}$, $\beta$ over $\{5000, 10000, 20000 \}$ and $\kappa$ over $\{0.5, 0.65, 0.75, 0.85, 0.95 \}$. (The range of the grid values was chosen so that it covered the recommended hyperparameter values in the literature.) For every combination of $(\gamma, \delta, \beta, \delta, \kappa)$ and $\alpha$ in $\{0.05, 0.01\}$, we extracted the adjacency matrix and estimated the number of blocks $\hat{K}$ using the SMT procedure. Finally, we estimate $(\hat{\gamma}, \hat{\delta}, \hat{\beta},\hat{\kappa}, \hat{\alpha})$ as those hyperparameter value that maximized the joint log likelihood of adjacency matrix under the assumptions of SBM. As given in the previous paragraph, this gives us the optimum value of $(2000, .25, 20000, 0.65, 0.05)$. For the optimized hyperparameter, the SMT procedure estimated the number of communities as $17$. Using the estimated communities, we plotted the tsne plot for the estimated clusters to get visual confirmation see the figure \ref{F:tsne}.

For each of the $17$ estimated cell clusters, we ran a binomial test to identify gene markers. We tested whether a particular gene is more expressed in a particular cell type compared to the other cell type. This gave a list of gene marker which is included in the Supplement. The obtained gene markers are crucial for understanding the cell type and its function. The $17$ clusters were divided into rod and cone bipolar cells based on the presence or absence of \emph{PRKCA} ( Kim et al. (2008) \cite{kim2008}). We found that the PRKCA gene present in the cluster numbered $8$, $15$ and $17$ with the cluster $8$ being considerably smaller. The rest of the bipolar cells were further divided into ON ($1$, $2$, $3$, $6$, $9$, $10$, $14$) and OFF ($4$, $5$, $7$, $11$, $12$, $16$) cells based on the ON bipolar markers \emph{Isl1} and/or \emph{Grm6} (Elshatory et al. (2007) \cite{elshatory2007}.

 We mapped the set of $17$ clusters to Shekhar et al. (2016) \cite{shekhar2016}'s classification of bipolar cells in mouse retina where they had a total of $15$ cells. Using the gene markers in \cite{shekhar2016}, our estimated $6$ bipolar OFF cells $4$, $5$, $16$, $12$ maps to $Bc3b$, $Bc2$, $Bc3a$, $Bc4$, respectively while bipolar OFF cells $7$ and $11$ map to $Bc3b$ and $Bc1b$ respectively. For bipolar ON cells, we find that our subgroups type $2$ and $9$ map to $Bc6$ and $Bc7$ while $10$, $6$ and $14$ map to $Bc5A$, $Bc 5b$ and $Bc5d$ respectively. Additionally, we find that ON subgroups $13$ map to $Bc5c$ and subgroups $1$ and $3$ map to $Bc8$ and $Bc9$. 

We found three rod bipolar cells $15$, $17$, and a smaller cluster $8$. It has been suggested that based on the current-voltage relationship that the rod bipolar cells have two functional classes Schilardi et al. (2022) \cite{schilardi2022}. Despite these differences, biologists agree on the presence of only one rod bipolar class as was shown in the study of mouse retina. Upon further investigation, we found some difference in their compositions. We found that CABP5 was predominantly present in the rod cluster subgroup $17$ but absent in the rod cluster subgroup $17$. CaBP5 can modulate voltage-gated calcium channels, 5 6 7 TRP channels, and inositol 1,4,5-trisphosphate (IP3) receptor, see Rieke et al. (2008) \cite{rieke2008}. The Figure \ref{F:tsne} suggests that the two subgroups are adjacent.

Curiously, we found a small cluster $8$ of a rod bipolar cells which is widespread in the Figure \ref{F:tsne}. We found that this cell cluster has PDE6A and CNGA1 genes that are strongly expressed.  The PDE6A gene encodes a subunit of the cGMP phosphodiesterase enzyme crucial for the function of rod photo-receptor cells whereas CNGA1 is a subunit of the photoreceptor cyclic-nucleotide gated (CNG) channel which converts light into electrical signals. A degeneration in PDE6A and CNGA1 could result in the loss of rod photo-receptor cells which can have a secondary damage on the bipolar cone cells and cause Retinitis Pigmentosa (RP). Recently, PDE6A gene is targeted for treatment in dogs for the restoration of vision in animals, \cite{mowat2017} while CNGB1 is targeted for mouse for restoration of vision \cite{CNGB1}.

Since our bipolar data comes from donors with advanced age, we hypothesize that cluster $8$ could give us insight into the progressive degradation of rod bipolar cells. Upon further investigation, we found other highly expressed gene in the cluster $8$ such as GNAT1, AIPL1, RHO, SAG and PDE6G (together or in combination) whose deterioration or mutation could indirectly affect or advance age-related-macular degenertion (AMD). For instance, the SAG gene encodes a protein S-antigen that regulates the light sensing process by coupling of rhodopsin to transducing G protein. Any mutation or degradation of SAG and/or RHO gene might result in Oguchi disease (night blindness) or recessive form of Retinitis Pigmentosa (RP). For instance, a study established that some mutation of GNAT1 can result in recessive form of RP, see \cite{GNAT1}. Also the absence of AIPL1 reduced the levels of cGMP phosphodiesterase subunits could cause rapid degeneration in rapid photoreceptor degeneration, see \cite{AIPL1}. However, this is empirically shown but we do not understand the underlying mechanism. We also found that RTBDN is also highly expressed in the cluster $8$ whose deficiency could disrupt metabolic homeostasis, increase the oxidative stress and accelerate the retinal degeneration see \cite{RTBDN}. The above empirical observations suggest that highly expressed genes of cluster 8 and their associations could be useful for tracking retinal degeneration and potentially could be targeted to restore vision.


\section{Discussion}\label{discussion}
We proposed a new approach for estimating the number of cell clusters of a network. We showed that our method is consistent for a moderate regime and because it relies on the second eigenvalue it is computationally cheaper. We highlighted that scRNA-seq datasets tend to be noisy and illustrated that the existing methods tend to fare bad for slightly moderate out-in ratios. Moreover, we highlighted the challenge in extracting networks out of noisy scRNA-seq datasets. To this end, we proposed a likelihood-based method for extracting networks out of scRNA-seq data. Unlike existing methods, our method automatically chooses the set of optimum hyper-parameters maximizing the underlying likelihood.

Finally, we applied our method to the data analysis of six benchmark datasets. We note that our method is slightly better than SC3, the pick of the methods. We also apply our method to the analysis of human bipolar retina cells. Here, as well, we did not choose our hyper-parameters but rather found the best set of hyper-parameters by maximizing the underlying likelihood. We estimated the cell cluster to be $17$ and investigated some of the major gene-markers of the bipolar ON and OFF cells. We underline that in above instances, we choose hyper-parameters as a part of grid-search maximizing the likelihood unlike the rest of the methods. This makes our approach principled and very useful for the analysis of scRNA-seq datasets.
In our analysis, we were able to map the BC-subtypes as given in \cite{shekhar2016}. Curiously, we found an extra small cluster $8$, a part of rod bipolar cells, which is perhaps not a subtype of bipolar cells but caused by the reorganization of the retina cells (due to age) in our donors. We discovered that mutations and absence of the highly expressed genes of cluster $8$ could result in various AMD diseases. Studying these highly expressed genes could be useful for potentially targeting in restoring vision or monitoring retinal degeneration.

\begin{figure}
    \centering
\includegraphics[width = 4 in]{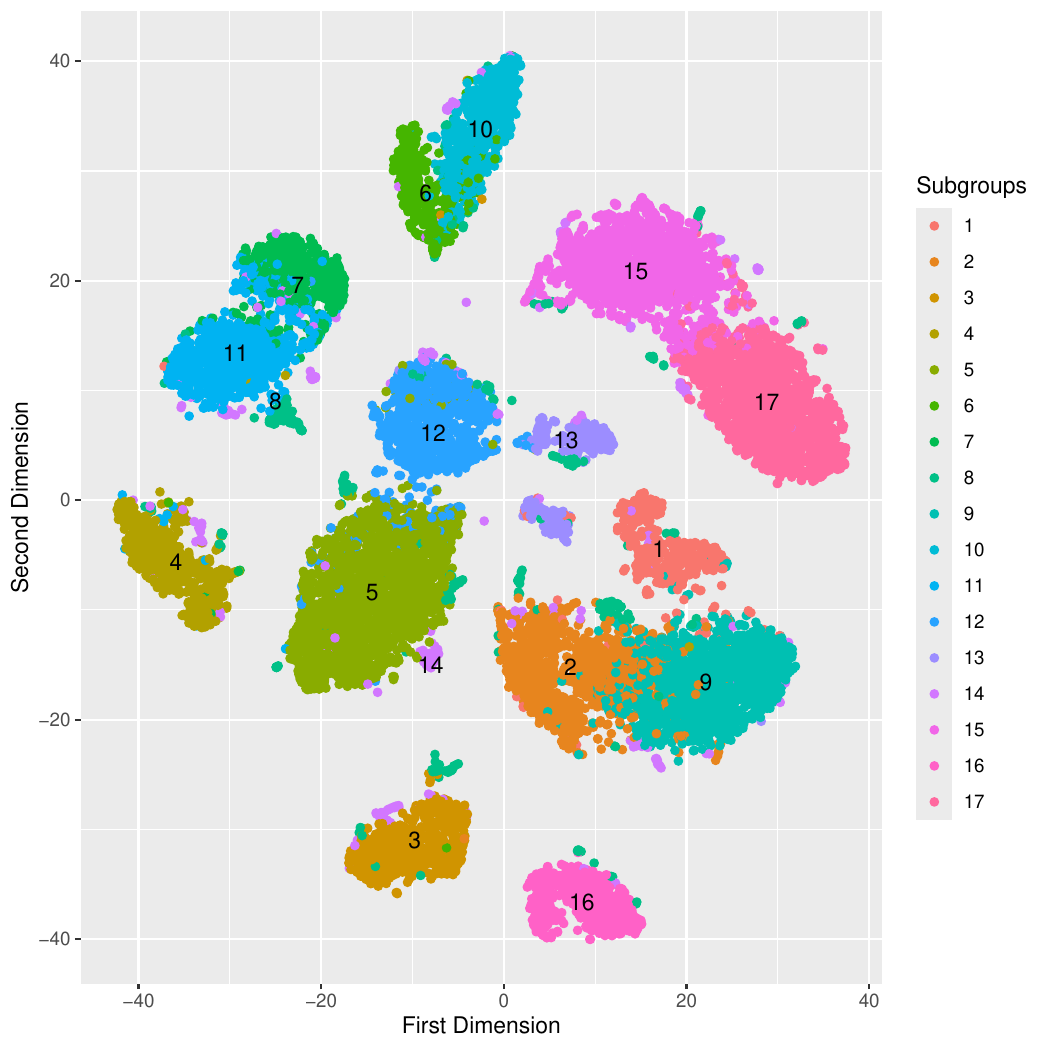}
\caption{The figure draws the tsne plot for the $17$ estimated clusters of the human retina bipolar cells. Here, the scRNA-seq network was extracted using the likelihood method while optimizing for the hyper-parameters and the estimated number of cluster was estimated using the SMT method.}\label{F:tsne}
\end{figure}

\bibliography{reference}

@article{lei16,
author={Lei, Jing},
title={A Goodness-of-fit Test for Stochastic Block Models},
journal={The Annals of Statistics},
volume={44},
number={1},
pages={401--424},
year={2016}
}

@article{zhao12,
author={Zhao, Yunpeng and Levina, Elizaveta and Zhu, Ji},
title={Consistency of Community Detection in Networks under Degree-Corrected Stochastic Block Models},
journal={The Annals of Statistics},
volume={40},
number={4},
year={2012},
pages={2266--2292}
}

@article{wang17,
author={Wang, Rachel Y. X. and Bickel, Peter J.},
title={Likelihood-based model selection for stochastic block models},
journal={The Annals of Statistics},
volume={45},
number={2},
year={2017},
pages={500--528}
}

@article{rohe11,
author={Rohe, Karl and Chatterjee, Sourav and Yu, Bin},
title={Spectral Clustering and The High-Dimensional Stochastic Blockmodel},
journal={The Annals of Statistics},
volume={39},
number={4},
year={2011},
pages={1878--1915}
}

@article{holland83,
author={Holland, P. W. and Laskey, K. B. and Leinhardt, S.},
title={Stochastic Block Models: First steps},
journal={Social Networks},
year={1983},
volume={5},
pages={109--137}
}

@article{newman04,
author={Newman, M. E. J. and Girvan, M.},
title={Finding and evaluating community structures in networks},
journal={Physical Review E},
year={2004},
volume={69},
pages={026113}
}

@article{zhao17,
author={Zhao, Yunpeng},
title={A survey on theoretical advances of community detection in networks},
journal={WIREs Comput Stat},
year={2017},
volume={9},
pages={c1403}
}

@article{bickel09,
author={Bickel, Peter J. and Chen, Aiyou},
title={A nonparametric view of network models and Newman-Girvan and other modularities},
journal={Proceedings of the National Academy of Sciences of the United States of America},
year={2009},
volume={106},
number={50},
pages={21068--21073}
}

@article{lei2015,
author={Lei, Jing and Rinaldo, Alessandro},
title={Consistency of Spectral Clustering in Sparse Stochastic Block Models},
journal={Annals of Statistics},
year={2015},
volume={43},
number={1},
pages={215--237}
}

@article{airoldi2008,
author={Airoldi, Edoardo M. and Blei, David M. and Fienberg, Stephen E. and Xing, Eric P.},
title={Mixed Membership Stochastic BlockModels},
journal={Journal of Machine Learning Research},
year={2008},
volume={9},
pages={1981--2004}
}

@article{karrer2011,
author={Karrer, Brian and Newman, E. J.},
title={Stochastic blockmodels and community structure in networks},
journal={Physics Review E},
volume={83},
number={1},
year={2011}
}

@article{amini2013,
author={Amini, Arash A. and Chen, Aiyou and Bickel, Peter J. and Levina, Elizaveta},
title={Pseudo-Likelihood methods for community detection in large sparse networks},
journal={Annals of Statistics},
volume={41},
number={4},
year={2013}
}

@article{choi2012,
author={Choi, D. S. and Wolfe, P. J. and Airoldi, E. M.},
title={Stochastic blockmodels with a growing number of classes},
journal={Biometrika},
volume={99},
number={2},
year={2012},
pages={273--284}
}

@article{joseph2016,
author={Joseph, Antony and Yu, Bin},
title={Impact of regularization on spectral clustering},
journal={Annals of Statistics},
volume={44},
number={4},
year={2016},
pages={1765--1791}
}

@article{ma2019,
author={Ma, Shujie and Su, Liangjun and Zhang, Yichong},
title={Determining the number of communities in degree-corrected stochastic block models},
journal={The Journal of Machine Learning and Research},
volume={22},
pages={3217-3279},
year={2021}
}

@article{chen2018,
author={Chen, Kehui and Lei, Jing},
title={Network Cross-Validation for Determining the number of communities in Network Data},
journal={Journal of the American Statistical Association},
volume={113},
number={521},
year={2018}
}

@article{lei2014,
author={Lei, J and Zhu, L},
title={A generic sample splitting approach for refined community recovery in stochastic block models},
journal={Statistical Sinica},
volume={27},
number={4},
year={2017}
}

@article{amini2018,
author={Amini, Arash A. and Levina, Elizaveta},
title={On semidefinite relaxations for the block model},
journal={Annals of Statistics},
volume={46},
number={1},
year={2018},
pages={149--179}
}

@article{chaudhuri2012,
author={Chaudhuri, Kamalika and Chung, Fan and Tsiatas, Alexander},
title={Spectral Clustering of Graphs with General Degrees in the Extended Partition Model},
journal={JMLR : Workshop and Conference Proceedings vol},
year={2012},
pages={35.1--35.23}
}

@article{shekhar2016,
author={Shekhar, Karthik and Lapan, Sylvian W. and Whitney, Irene E. and Tran, Nicholas M. et al.},
title={Comprehensive Classification of Retinal Bipolar Neurons by Single-Cell Transcriptomics},
journal={Cell},
year={2016},
volume={166},
number={5},
pages={1308--1323}
}

@article{satija2015,
author={Satija, Rahul and Farrell, Jeffrey A. and Gennert, David and Schier, Alexander F and Regev, Aviv},
title={Spatial reconstruction of single-cell gene expression data},
journal={Nature Biotechnology},
year={2015},
volume={33},
number={5},
pages={495--502}
}

@article{kiselev2017,
author={Kiselev, Valdimir Yu. and Kristina, Kirschner and Schaub, Michael T. and Andrews, Taullulah et al.},
title={SC3- consensus clustering of single-cell RNA-Seq data},
journal={Nature Methods},
year={2017},
volume={14},
number={5},
pages={483--486}
}

@article{blondel2008,
author={Blondel, Vincent D and Guillaume, Jeane-Loup and Lambiotte, Renaud and Lefebvre, Etienne},
title={Fast unfolding of communities in large networks},
journal={Journal of Statistical Mechanics Theory and Experiment},
year={2008},
page={10008}
}

@article{lyu2019,
author={Lyu, Yafei et al.},
title={Implication of specific retinal cell-type involvement and gene expression changes in AMD progression using integrative analysis of single-cell and bulk RNA-seq profiling},
journal={Scientific Reports},
volume={11},
number={15612},
year={2021}
}

@article{leeTWCov16,
author={Lee, Ji Oon and Schnelli, Kevin},
title={Tracy-Widom distribution for the largest eigenvalue of real sample covariance matrices with general population},
journal={Annals of Applied Probability},
volume={26},
number={6},
year={2016}
}

@article{altschuler2010,
author ={Altschuler, S. and Wu, L.},
title={Cellular heterogeneity: do differences make a difference?},
journal={Cell},
volume={141},
number={4},
pages={559-63},
year={2010}
}

@article{heumos2023,
author={Heumos, L. and Schaar, A.C. and Lance, C. et al.},
title={Best practices for single-cell analysis across modalities},
journal={Nat Rev Genet},
volume={24}, 
pages={550–572},
year=2023
}

@article{sengupta2015,
    author = {Sengupta, Srijan and Chen, Yuguo},
    title = {Spectral Clustering in Heterogeneous Networks} ,
    journal = {Statistical Sinica},
    volume = 25,
    number =3,
    year = 2015
}

@article{qin2013,
author={Qin, Tai and Rohe, Karl},
title={Regularized Spectral Clustering Under the Degree-Corrected Stochastic Block Model},
journal={Proceedings of the Neural Information Processing Systems},
volume={2},
pages={3120--3128},
year={2013}
}

@article{jin2023,
author={Jin, Jiashun and Tracy, Zheng Ke and Luo, Shenming and Wang, Minzhe},
title={Optimal Estimation of the Number of Network Communities},
journal={Journal of the American Statistical Association},
volume={118},
number={543},
year={2023}
}

@article{li2020,
author={Li, Tianxi and Levina, Elizabeth and Zhu, Ji},
title={Network cross-validation by Edge Sampling},
journal={Biometrika},
volume={107},
number={2},
pages={257--276},
year={2020}
}

@article{yan2018,
author={Yan, Bowei and Sarkar, Purnamitra and Cheng, Xiuyuan},
title={Provable Estimation of the Number of Blocks in Block Models},
journal={Proceedings of the Twenty-First International Conference on Artificial Intelligence and Statistics},
volume={84},
pages={1185--1194},
year={2018}
}

@article{cerqueira2020,
author={Cerqueira, A and Leonardi, F.},
title={Estimation of the Number of Communities in the Stochastic Block Model},
journal={IEEE Transactions on Information Theory},
volume={66},
number={10},
pages={6403--6412},
year={2020}
}

@article{cerqueira2023,
author={Cerqueira, A and Gallo, S and Leonardi, F and Vera, C},
title={Consistent Model Selection For the Degree Corrected Stochastic Block Model},
journal={ALEA, Lat. Am. J. Probab. Math. Stat},
volume={21},
pages={267-292},
year={2024}
}

@article{chen2021,
author={Chen, Fan and Roch, Sebastien and Rohe, Karl and Yu, Shuqi},
title={Estimating Graph Dimension With Cross-Validated Eigenvalues},
journal={https://arxiv.org/abs/2108.03336},
year={2021}
}

@article{le2022,
author={Le, Can and Levina, Elizabeth},
title={Estimating the Number of Communities By Spectral Methods},
journal={Electronic Journal of Statistics},
volume={16},
pages={3315--3342},
year={2022}
}

@article{edgar2002,
author={Edagr, R \. and Domrachev, M \. and A \. E \., L \.},
title={Gene Expression Omnibus: Ncbi gene expression and hybridization array data recovery},
journal={Nucleic Acid Res.},
volume={30},
number={1},
pages={207--210},
year={2002}
}

@article{kim2008,
author={Kim et al. },
title={Identification of molecular markers of the bipolar cells in the muraine retina},
journal={Journal of Comparative Neurology},
volume={507},
number={5},
pages={1795--1810},
year={2008}
}

@article{elshatory2007,
author={Elshatory, Y. and  Everhart, D. and Deng, M. and Xie, X. and Barlow, R. and Gan, L.},
title={Islet-1 controls the differentiation of retinal bipolar and cholinergic amacrine cells},
journal={ J Neurosci.},
volume={27},
number={6},
pages={12707--20},
year={2007}
}

@article{schilardi2022,
author={Schilardi, Giulia and Kleinlogel, Sonja },
title={Two Functional Classes of Rod Bipolar Cells in the Healthy and Degenerated Optogenetically Treated Murine Retina},
journal={Front Cell Neurosci.},
volume={2022 Jan 13},
number={15},
year={2022}
}

@article{irizarry2024,
author={Grabski, Isabella and Street, Kelly and Irizarry, Rafael A.},
title={Significance Analysis for Clustering with single-cell RNA-sequencing data},
journal={Nature Methods},
volume={20},
pages={1196-1202},
year={2023}
}

@article{rieke2008,
author={Rieke, Fred and Lee, Amy},
title={Characterization of Ca2+-Binding Protein 5 Knockout Mouse Retina},
journal={Investigative Ophthalmology and Visual Science},
year={2008}, 
volume={49}, 
pages={5126-5135}
}

@article{sischika2025,
author={Sischika, Benjamin and Kauremann, Goran},
title={Stochastic Block Smooth Graphon Model},
journal={Journal of Compuation and Graphical Statistics},
volume={34},
number={1},
pages={140-154},
year={2025}
}

@article{mowat2017,
author={Mowat, Freya and Occelli, Laurence and Bartoe, Joshua and Gervais, Kristen and Bruewer, Ashlee},
title={Gene Therapy in a Large Animal Model of PDE6A-Retinitis Pigmentosa},
journal={Front Neurosci.},
year={2017},
volume={Jun 20},
number={11}
}

@article{CNGB1,
title={Gene therapy restores vision and delays degeneration in the CNGB1 mouse model of retinitis pigmentosa},
author={Koch, Susanne and Sothilingam, Vithiyanjali and Garrido, Marina et al.},
journal={Human Molecular Genetics},
volume={21},
number={20},
year={2012}
}

@article{AIPL1,
author={ Kolandaivelu, S. and  Huang, Jing and Hurley, James B. and Ramamurthy, Visvanathan},
title={AIPL1, a Protein Associated with Childhood Blindness, Interacts with Subunit of Rod Phosphodiesterase (PDE6) and Is Essential for Its Proper Assembly},
journal={J Biol Chem},
volume={284},
number={45},
year={2009}
}

@article{GNAT1,
author={
Carrigan, M. and  Duignan, E. and Humphries, P. and Palfi, A. and Kenna, P. and Farrar, G.},
title={A novel homozygous truncating GNAT1 mutation implicated in retinal degeneration},
journal={ Br. J Ophthalmol},
year={2016},
volume={100},
number={4},
pages={495--500}
}

@article{RTBDN,
author={Zhao, X. and Tebbe, L. and Naash, MI and Al-Ubaidi, MR.},
title={The Neuroprotective Role of Retbindin, a Metabolic Regulator in the Neural Retina},
journal={Front Pharmacol.},
year={2022 },
volume={Jul 6},
number={13}
}

@article{le2024,
author={Zhixuan, Shao and Le, Can M.},
title={Determining the Number of Communities in Sparse and Imbalanced Settings},
journal={https://arxiv.org/abs/2406.04423},
year={2024}}
\end{document}